\documentclass[conference]{IEEEtran}

\usepackage{cite}
\usepackage{amsbsy,amsmath,amssymb,epsfig,bbm,mathrsfs,multirow,amsthm,amsfonts}
\usepackage{algorithm,algorithmic}
\usepackage{graphicx}
\usepackage{textcomp}
\usepackage{multicol}
\usepackage{url}

\usepackage[table]{xcolor}

\PassOptionsToPackage{hyphens}{url}

\def\BibTeX{{\rm B\kern-.05em{\sc i\kern-.025em b}\kern-.08em
    T\kern-.1667em\lower.7ex\hbox{E}\kern-.125emX}}

\newtheorem{proposition}{Proposition}

\newtheorem{definition}{Definition}

\usepackage[T1]{fontenc}
\usepackage{makecell}
\usepackage{amstext} 
\usepackage{subcaption} 
\definecolor{mygray}{gray}{0.6}
\definecolor{myblue}{rgb}{0.8,0.85,1}

\usepackage{array, tabularx, boldline}
\newcolumntype{L}[1]{>{\raggedright\let\newline\\\arraybackslash\hspace{0pt}}m{#1}}
\newcolumntype{C}[1]{>{\centering\let\newline\\\arraybackslash\hspace{0pt}}m{#1}}
\newcolumntype{R}[1]{>{\raggedleft\let\newline\\\arraybackslash\hspace{0pt}}m{#1}}

\usepackage{cellspace}
\begin{document}

\title{Collaborative Coded Computation Offloading: An All-pay Auction Approach}


\author{Jer Shyuan Ng$^{1,2}$, Wei Yang Bryan Lim$^{1,2}$, Sahil Garg$^{3}$, Zehui Xiong$^{2,4}$,\\ Dusit Niyato$^{4}$, Mohsen Guizani$^{5}$, and Cyril Leung$^{6,7}$\\
$^1$Alibaba Group~$^2$Alibaba-NTU JRI~$^3$École de Technologie Supérieure~$^4$SCSE, NTU, Singapore~$^5$Qatar University~\\$^6$LILY Research Center, NTU, Singapore~$^7$ECE, UBC, Canada \vspace*{-5mm}}

\maketitle

\begin{abstract} 

As the amount of data collected for crowdsensing applications increases rapidly due to improved sensing capabilities and the increasing number of Internet of Things (IoT) devices, the cloud server is no longer able to handle the large-scale datasets individually. Given the improved computational capabilities of the edge devices, coded distributed computing has become a promising approach given that it allows computation tasks to be carried out in a distributed manner while mitigating straggler effects, which often account for the long overall completion times. Specifically, by using polynomial codes, computed results from only a subset of devices are needed to reconstruct the final result. However, there is no incentive for the edge devices to complete the computation tasks. In this paper, we present an all-pay auction to incentivize the edge devices to participate in the coded computation tasks. In this auction, the bids of the edge devices are represented by the allocation of their Central Processing Unit (CPU) power to the computation tasks. All edge devices submit their bids regardless of whether they win or lose in the auction. The all-pay auction is designed to maximize the utility of the cloud server by determining the reward allocation to the winners. Simulation results show that the edge devices are incentivized to allocate more CPU power when multiple rewards are offered instead of a single reward.

\end{abstract}

\begin{IEEEkeywords}
Coded distributed computing, straggler effects mitigation, all-pay auction, Bayesian Nash equilibrium
\end{IEEEkeywords}

\section{Introduction}

The increasing sensing capabilities of Internet of Things (IoT) devices such as smart meters and smartphones, coupled with reliable wireless communication technologies, have enabled the development of crowdsensing applications, which leverage large-scale data for analysis. The IoT devices upload the data collected by the on-board sensors to the cloud server for analysis. Through deep learning techniques, the cloud server aims to build accurate decision-making models for smart city applications, such as estimation of road conditions~\cite{kalim2016crater}, air quality monitoring~\cite{dutta2017towards} and tracking of medical conditions~\cite{jovanovic2019mobile}. 

However, as the number of IoT devices increases rapidly, the cloud server may not be able to handle the ever-expanding datasets individually. Edge computing~\cite{shi2016promise} has become a promising approach that extends cloud computing services to the edge of the networks. Specifically, by leveraging on the computational capabilities, e.g., Central Processing Unit (CPU) power, of the edge devices, e.g., base stations, the cloud server can offload its computation tasks to the edge devices. Instead of offloading the entire computation to a single edge device, the cloud server can divide the large datasets into smaller sets and distribute them to the edge devices for processing. This alleviates the bottleneck and single point of failure problems. Upon completing the allocated computation subtasks, the edge devices return the computed results to the cloud server for aggregation to reconstruct the final result. 

To perform the distributed computation tasks efficiently, coded distributed computing~\cite{ng2020survey} is being increasingly used to mitigate the straggler effects caused by factors such as imbalanced work allocation, contention of shared resources and network congestion~\cite{dean2013tail,ananthanarayanan2010reining}.


In this paper, we consider that a cloud server, i.e., master, aims to build a prediction model in which the training process involves a large number of matrix multiplication computations. Matrix multiplication is an important operation underlying many data analytics applications, e.g., machine learning, scientific computing and graph processing~\cite{yu2017polynomial}. Due to a lack of computational capabilities and storage resources, the resource-constrained master employs edge devices, i.e., workers, to facilitate the distributed computation task. First, the master uses polynomial codes~\cite{yu2017polynomial} to determine the algebraic structure of the encoded submatrices and distributes them to the workers. Then, the workers perform local computations based on their allocated datasets, and return the computed results to the master. The master reconstructs the final result by using fast polynomial interpolation algorithms such as the Reed-Solomon decoding algorithm~\cite{didier2009efficient}. 

However, there is no incentive for the workers to participate in or to complete their allocated distributed computation tasks. To design an appropriate incentive mechanism, it is important to consider the unique aspect of the coded distributed computing framework. Specifically, even though workers are allocated a subset of the entire dataset for computations, some of the workers may perform computations without any compensation if their computed results are not used. As such, we propose an all-pay auction to improve worker participation. In traditional auctions such as first-price and second-price auctions, only the winners of the auctions pay. In contrast, in all-pay auctions, regardless of whether the bidders win or lose, they are required to pay. In this all-pay auction, the bids of the workers are represented by the CPU power, i.e., the number of CPU cycles, allocated by the workers to complete the computation tasks. In other words, the larger the CPU power allocated, the higher the bid of the worker. 

There are two advantages of an all-pay auction~\cite{luo2016incentive}. Firstly, it reduces the probability of non-completion of allocated tasks. This differs from traditional auctions in which the winners of the auctions can still choose not to complete their tasks and give up the reward that is promised by the auctioneer. As a result, the auctioneer needs to conduct another round of auction to meet its objectives. Secondly, it reduces the coordination cost between the auctioneer and the bidders. Specifically, in traditional auctions, the participants need to bid then contribute whereas in all-pay auctions, the bids of the participants are directly determined by their contributions. In other words, the participants do not need to bid explicitly in all-pay auctions.

\section{System Model}
\label{sec:system}

We consider a master-worker setup in a heterogeneous distributed computing network. The system model consists of a master, e.g., cloud server and a set $\mathcal{I}=\{1,\ldots,i,\ldots,I\}$ of $I$ workers, e.g., edge devices such as base stations, which have different computational capabilities. In crowdsensing, for example, the ubiquity of the IoT devices as well as their on-board sensing and processing capabilities are leveraged to collect data for many innovative crowdsensing applications, e.g., weather prediction, location-based recommendation and forensic analysis. Given the large amount of sensor data collected from different IoT devices, the master aims to perform training of the machine learning algorithm to complete a user-defined task. As the number of IoT devices that contribute to the crowdsensing task increases, so does the size of the dataset that the master needs to handle. However, the master may not have sufficient resources, i.e., computation power to handle the growing dataset. Instead, it may utilize the resources of the workers to collaboratively complete the distributed computation tasks.

One of the main challenges in performing distributed computation tasks is the straggler effects. Coding techniques such as polynomial codes \cite{yu2017polynomial} can be used to mitigate straggler effects by reducing the recovery threshold, i.e., the number of workers that need to submit their results for the master to reconstruct the final result. In order to perform distributed coded matrix multiplication computations, i.e., $\mathbf{C}=\mathbf{A}^\top \mathbf{B}$ where $\mathbf{A}$ and $\mathbf{B}$ are input matrices, $\mathbf{A} \in \mathbb{F}_q^{s\times r}$ and $\mathbf{B} \in \mathbb{F}_q^{s\times t}$ for integers $s$, $r$, and $t$ and a sufficiently large finite field $\mathbb{F}_q$, there are four important steps:

\begin{enumerate}
\item \emph{Task Allocation: }Given that all workers are able to store up to $\frac{1}{m}$ fraction of matrix $\mathbf{A}$ and $\frac{1}{n}$ fraction of matrix $\mathbf{B}$, the master divides the input matrices into submatrices $\mathbf{\tilde{A}}_i=f_i(\mathbf{A})$ and $\mathbf{\tilde{B}}_i=g_i(\mathbf{B})$, $\forall i \in \mathcal{I}$, where $\mathbf{\tilde{A}}_i \in \mathbb{F}_q^{s\times \frac{r}{m}}$ and $\mathbf{\tilde{B}}_i \in \mathbb{F}_q^{t\times \frac{r}{n}}$ respectively. Specifically, $\mathbf{f}$ and $\mathbf{g}$ represent the vectors of functions such that $\mathbf{f}=(f_1,\ldots,f_i,\ldots,f_I)$ and $\mathbf{g}=(g_1,\ldots,g_i,\ldots,g_I)$, respectively. Then, the master distributes the submatrices to the workers over the wireless channels for computations.

\item \emph{Local Computation: }Each worker $i$ is allocated submatrices $\mathbf{\tilde{A}}_i$ and $\mathbf{\tilde{B}}_i$ by the master. Based on the allocated submatrices, the workers perform matrix multiplication, i.e., $\mathbf{\tilde{A}}_i^\top \mathbf{\tilde{B}}_i$, $\forall i \in \mathcal{I}$.

\item \emph{Wireless Transmission: }Upon completion of the local computations, the workers transmit the computed results, i.e., $\mathbf{\tilde{C}}_i=\mathbf{\tilde{A}}_i^\top \mathbf{\tilde{B}}_i$ to the master over the wireless communication channels.

\item \emph{Reconstruction of Final Result: }By using polynomial codes, the master is able to reconstruct the final result upon receiving $K$ out of $I$ computed results by using decoding functions. In other words, the master does not need to wait for all $I$ workers to complete the computation task. Note that although there is no constraint on the decoding functions to be used with polynomial codes, a low-complexity decoding function ensures the efficiency of the overall matrix multiplication computations.
\end{enumerate}

In our system model, we consider the computations for distributed matrix multiplication. However, there are other types of distributed computation problems, e.g., stochastic gradient descent, convolution and Fourier transform. The matrix multiplication may also involve more than two matrices. Our system model can be easily extended to solve the matrix multiplication of more than two matrices and different computation problems.

By using polynomial codes \cite{yu2017polynomial}, the optimum recovery threshold that can be achieved where each worker is able to store up to $\frac{1}{m}$ of matrix $\mathbf{A}$ and $\frac{1}{n}$ of matrix $\mathbf{B}$ is defined as:
\begin{equation}
K=mn.
\end{equation}

However, there is no incentive for the workers to be one of the $K$ workers to complete their local computations and return their computed results to the master. The workers may not want to participate in the distributed computation tasks or perform to the best of their abilities. For example, the workers may not allocate large CPU power for the coded distributed tasks, resulting in longer time needed to complete the computation tasks. In order to incentivize the workers to allocate more CPU power to complete the allocated computation tasks, we adopt an all-pay auction approach to determine the rewards to the workers while maximizing the utility of the master.

\subsection{Utility of the Master}

Given that the master only needs the computed results from $K$ workers to reconstruct the final result, the master offers $K$ rewards, represented by the set $\mathcal{K}=\{1,\ldots, k,\ldots, K\}$ where $K \leq I$. Specifically, there are $K$ rewards for which $I$ workers compete. Since only $K$ rewards are offered, $I-K$ workers do not receive any reward from the master, even though they perform the matrix multiplication computations given the allocated submatrices.

The size of reward $k$ is represented by $M_k$. The worker that allocates larger CPU power is offered larger reward. In particular, the worker that allocates the largest amount of CPU power receives a reward of $M_1$, the worker with the second largest allocation receives reward $M_2$ and the worker with the $k$-th largest allocation of CPU power is offered reward $M_k$. If two or more workers allocate the same amount of CPU power to perform the distributed computation tasks, ties will be randomly broken. In other words, if both workers are ranked $k$, one is ranked $k$ and the other is ranked $k+1$. Hence, without loss of generality, $M_1\geq M_2 \geq \cdots \geq M_K > 0$. The total amount of reward offered by the master is denoted by $\sigma$, i.e., $\sigma=\sum_{k=1}^{K}M_{k}$. The aim of the master is to spend the entire fixed reward to maximize the CPU power allocated by the workers. 

As such, the expected utility of the master, $\pi$ is expressed as follows:
\begin{equation}
\label{profit}
\pi=\mathbb{E}[z_{1:I}+z_{2:I}+\cdots+z_{K:I}-\sigma],
\end{equation}
where $z_{k:I}$ represents the order statistics of the worker's CPU power allocation. Specifically, $z_{1:I}$ and $z_{k:I}$ denote the highest and $k$-th highest CPU power allocation respectively among $I$ workers.

\subsection{Utility of the Worker}

To perform the local computations on the allocated submatrices, each worker $i$ consumes computational energy, $e_i$, which is defined as:
\begin{equation}
e_i=\kappa a_{i}(z_i)^2,
\end{equation} 
where $\kappa$ is the effective switch coefficient that depends on the chip architecture \cite{zhang2018energy}, $a_{i}$ is the total number of CPU cycles required to complete the allocated computation subtask and $z_i$ is the CPU power allocated by worker $i$ for the computation subtask. By using the polynomial codes, the computation task is evenly partitioned and distributed among all workers. As a result, the total number of CPU cycles that are needed to complete the allocated computation tasks is the same for all workers, i.e., $a_{i}=a, \; \forall i \in \mathcal{I}$. The unit cost of computational energy incurred by worker $i$, $\forall i \in\mathcal{I}$ is denoted by $\theta$, where the unit cost of computational energy is the same for all workers.


Each worker $i$ has a valuation $v_{i}$ for the total reward $\sigma$. The workers' valuations, $v_i$,  $\forall i \in \mathcal{I}$ are independently drawn from $v_i \in [\underline{v},\bar{v}]$ such that $\underline{v}$ and $\bar{v}$ are strictly positive based on $F(v)$, where $F(v)$ is the cumulative distribution function (CDF) of $v$. The total cost of worker $i$ is represented by $\theta e_{i}$. As a result, the utility of worker $i$ for winning reward $M_{k}$, $\forall k \in \mathcal{K}$, is expressed as:
\begin{equation}
\label{eqn:payoff}
\alpha_{i}=
	\begin{cases}
   	v_{i}M_{k}-\theta e_{i}, & \text{if worker $i$ wins $M_{k}$ reward}, \\
   	-\theta e_{i}, & \text{otherwise}.
 	\end{cases}
\end{equation}

\section{All-pay Auction}

In this section, we present the design of an all-pay auction that incentivizes the workers to allocate more CPU power for the allocated computation subtasks. In this auction, the master is the auctioneer whereas the workers are the bidders. The bid of a worker is represented by the CPU power that it allocates for its computation subtasks. In this all-pay auction, all $I$ workers, i.e., bidders, pay their bids regardless of whether they win or lose the auction. 

Each worker $i$ knows its own valuation, $v_{i}$ but does not know the valuation of any other worker, $i'\neq i$. This establishes a one-dimensional incomplete information setting. In addition, if each worker has a different unit cost of computational energy which is only known to itself, we consider the two-dimensional incomplete information setting. The dimension of private information can be reduced following the procedure in~\cite{yoon2012optimal}. In this work, we consider a one-dimensional incomplete information setting where the unit cost of computational energy is the same for all workers but the workers' valuations are heterogeneous and private. 

Given the utility of worker $i$, $\alpha_i$ in Equation~(\ref{eqn:payoff}), the objective of worker $i$ to maximize its expected utility, $u_{i}$, is defined as follows:
\begin{equation}
\label{eqn:utility}
\max_{z_i}u_{i}=v_{i}\sum_{k=1}^Kp_{i}^k{M_{k}}-\theta\kappa a(z_{i})^2,
\end{equation}
where $p_{i}^k$ is the winning probability of reward $M_k$ by worker~$i$.

Although the worker does not know exactly the valuations of other workers, it knows the distribution of the other workers' valuations based on past interactions, which is a common knowledge to all workers and the master. In our model, we consider that all the workers' valuations are drawn from the same distribution, which constitutes a symmetric Bayesian game where the prior is the distribution of the workers' valuations.

\begin{definition} \cite{tie2014optimal} A pure-strategy Bayesian Nash equilibrium is a strategy profile $\mathbf{z}^*=(z_1^*,\ldots,z_i^*,\ldots,z_I^*)$ that satisfies
$$u_{i}(z_{i}^*,\mathbf{z}_{i'}^*) \geq u_{i}(z_{i},\mathbf{z}_{i'}^*), \; \forall i \in \mathcal{I}. $$
\end{definition}
The subscript $i'$ represents the index of other workers other than worker $i$. Specifically, $\mathbf{z}_{i'}^*=(z_1^*,z_2^*,\ldots,z_{i-1}^*,z_{i+1}^*,\ldots,z_I^*)$ represents the equilibrium CPU power allocations of all other workers other than CPU power allocation of worker $i$. At Bayesian Nash equilibrium, given the belief of worker $i$, $\forall i \in \mathcal{I}$ about the valuations and that the CPU powers allocated by other workers, $i'$ where $i \neq i'$ are at equilibrium, $\mathbf{z}^*_{i'}$, worker $i$ aims to maximize its expected utility. 

\begin{proposition}Under incomplete information setting, the all-pay auction admits a pure-strategy Bayesian Nash equilibrium that is strictly monotonic where the bid of a worker strictly increases in its valuation.
\end{proposition}

\begin{proof}
The proof is omitted due to space constraints.
\end{proof}

Since the equilibrium CPU power allocation of worker~$i$, which is represented by $z_{i}^*$, is a strictly monotonically increasing function of its valuation $v_{i}^*$, we express the equilibrium strategy of worker $i$ as a function represented by $\beta(\cdot)$, i.e., $z_{i}^*=\beta_{i}(v_{i})$. Given the strict monotonicity, the inverse function also exists where $v_i(\cdot)=\beta_i^{-1}(\cdot)$ and it is an increasing function. Due to the incomplete information setting, the objective of worker $i$ to maximize its expected utility in Equation~(\ref{eqn:utility}) can be expressed as follows:
\begin{equation}
\max_{z_i}u_{i}=v_{i}\sum_{k=1}^Kp_{i}^k(z_{i},\beta_{i'}(v_{i'})){M_{k}}-c(\beta(v_{i})),
\end{equation}
where the cost of worker $i$ is represented by the function $c(\cdot)=\theta\kappa a(\beta(v_{i}))^2$.

Since the workers are symmetric, i.e., the valuations of workers are drawn from the same distribution, the symmetric equilibrium strategy for each worker $i$, $\forall i \in \mathcal{I}$ can be derived. We first assume that there are $I$ rewards, where $M_{1}\geq M_{2}\geq \cdots \geq M_{K}>M_{K+1}=M_{K+2}=\cdots=M_{I}=0$. The valuations of the workers, $v_1,\ldots, v_{i},\ldots, v_{I}$ are ranked and represented by its order statistics, which are expressed as $v_{1:I}\geq v_{2:I} \geq\cdots\geq v_{I:I}$. In particular, $v_{k:I}$ represents the $k$-th highest valuation among the $I$ valuations which are drawn from the common distribution $F(v)$. Given the order statistics of the workers' valuations, $\forall i \in \mathcal{I}$, the corresponding cumulative distribution function and probability density function are represented by $F_{k:I}$ and $f_{k:I}$ respectively. Similarly, when dealing with the valuations of all workers, other than that of worker $i$, the order statistic is represented by $v_{k:I-1}$, which represents the $k$-th highest valuation among the $I-1$ valuations. The corresponding cumulative distribution function and probability density function are represented by $F_{k:I-1}$ and $f_{k:I-1}$ respectively.

Given that other workers $i'$, where $i' \neq i$, follow a symmetric, increasing and differentiable equilibrium strategy $\beta(\cdot)$, worker $i$ will never choose to allocate a CPU power greater than the equilibrium strategy given the highest valuation. In other words, worker $i$ will never allocate $z_i>\beta(\bar{v})$. Besides, the optimal strategy of the worker with lowest valuation $\underline{v}$ is not to allocate any CPU power. On one hand, when the number of rewards offered is smaller than the number of workers, i.e., $K<I$, the worker with lowest valuation $\underline{v}$ will not win any reward. On the other hand, when the number of rewards offered is larger than or equal the number of workers, i.e., $K\geq I$, the worker with lowest valuation $\underline{v}$ will win a reward without allocating any CPU power. Hence, $u_i({\underline{v}})=0$. With this, the expected utility of worker $i$ with valuation $v_{i}$ and CPU power allocation $z_i=\beta(v_i)$ is expressed as follows:
\begin{equation}
\label{eqn:utilityprob}
u_{i}=v_{i}\sum_{k=1}^I[F_{k:I-1}(v_i)-F_{k-1:I-1}(v_i)]{M_{k}}-c(\beta(v_i)),
\end{equation}
since $M_{k+1}=\cdots=M_{I-1}=M_{I}=0$, $F_{0:I-1}(z_i)\equiv0$ and $F_{I:I-1}(z_i)\equiv1$. 

By differentiating Equation~(\ref{eqn:utilityprob}) with respect to the variable $w_i$ and equating the result to zero, we obtain the following:
\begin{equation}
\label{differentiate}
0=v_{i}\sum_{k=1}^I[f_{k:I-1}(v_i)-f_{k-1:I-1}(v_i)]{M_{k}}-c'(\beta(v_i))\beta'(v_i).
\end{equation}

When maximized, the marginal value of the reward is equivalent to the marginal cost of the CPU power. Since we have the differentiated function $c'(\cdot)$, the function $c(\cdot)$ can be found by using the integral of Equation~(\ref{differentiate}). At equilibrium, when the expected utility of worker $i$, $\forall i \in \mathcal{I}$, is maximized, we have the following:
\begin{equation}
\begin{split}
c(\beta(v_i))&=\sum_{k=1}^{I}M_{k}\int_{\underline{v}}^{v_i}v_{i}[f_{k:I-1}(v_{i})-f_{k-1:I-1}(v_{i})]dv_{i}\\
&=\sum_{k=1}^{I-1}(M_{k}-M_{k+1})\int_{\underline{v}}^{v_i}v_{i}f_{k:I-1}(v_{i})dv_{i}.
\end{split}
\end{equation}

Thus the equilibrium strategy for worker $i$ with valuation $v_{i}$, $\forall i \in \mathcal{I}$, is expressed as:
\begin{equation}
z_{i}^*=\beta(v_{i})=c^{-1}\left(\sum_{k=1}^{I-1}(M_{k}-M_{k+1})\int_{\underline{v}}^{v_i}v_{i}f_{k:I-1}(v_{i})dv_{i} \right).
\end{equation}

Given the equilibrium strategy of worker $i$, $\forall i \in \mathcal{I}$, the master aims to maximize its expected utility, $\pi$. By using the polynomial codes, the master is able to reconstruct the final result by using the computed results from $K$ workers. Since the master spends the fixed reward $\sigma$ completely, the maximization problem in Equation~(\ref{profit}) is equivalent to maximizing the allocation of CPU power, which is expressed as follows: 
\begin{multline}
\pi=\mathbb{E}[\beta(v_{1:I})+\beta(v_{2:I})+\cdots+\beta(v_{K:I})]\\
\shoveleft{=\sum_{i=1}^K\int_{\underline{v}}^{\bar{v}}\beta(v)dF_{i:I}(v)}\\
\shoveleft{=K\int_{\underline{v}}^{\bar{v}}\beta(v)dF(v)}\\
\shoveleft{=K\int_{\underline{v}}^{\bar{v}}c^{-1}\left(\sum_{k=1}^{I-1}(M_{k}-M_{k+1})\int_{\underline{v}}^{v}vf_{k:I-1}(v)dv \right)dF(v)}\\
\shoveleft{=K\int_{\underline{v}}^{\bar{v}}c^{-1}(\sum_{k=1}^{I}M_{k}\int_{\underline{v}}^{v}v[f_{k:I-1}(v)}\\
{ -f_{k-1:I-1}(v)dv])dF(v)}.
\end{multline}

Since the equilibrium strategy of worker $i$, $\forall i \in \mathcal{I}$, is affected by the reward structure, the master needs to determine the structure of the rewards such that it maximizes the CPU power allocation of the workers, thereby maximizing its own utility, $\pi$.

\section{Reward Structure}

Given that the master spends the total amount of the reward,~$\sigma$, the design of the optimal reward sequence is important to maximize the CPU power allocation of the workers since the equilibrium strategies of the workers depend on the differences between consecutive rewards.

The master needs to first decide whether to allocate the total amount of reward, $\sigma$ to only one winner, i.e., winner-take-all reward structure, or to split the reward into several smaller rewards.

\begin{proposition}

Given that the cost functions are convex, it is not optimal to only offer one reward where $M_1=\sigma$ and $ M_2=\cdots=M_K=\cdots=M_I=0$ since $\frac{\partial \pi}{\partial M_{k-1}}-\frac{\partial \pi}{\partial M_{k}} <0$, for $k=2,\ldots,I$. In particular, if $\frac{\partial \pi}{\partial M_{1}}-\frac{\partial \pi}{\partial M_{2}} <0$, it is not optimal to only offer a reward. 
\end{proposition}

\begin{proof}

Following the procedure in~\cite{yoon2012optimal}, we show that it is not optimal to offer a single reward given the cost functions of the workers are convex. The details of the proof are omitted due to space constraints.
\end{proof} 

Since the winner-take-all reward structure is not optimal, the master is better off offering multiple rewards. Given that $K$ rewards are offered, the master can consider several reward sequences such as (i) homogeneous reward sequence, (ii) arithmetic reward sequence and (iii) geometric reward sequence. Specifically, the reward sequence is expressed as follows:
\begin{itemize}
\item Homogeneous reward sequence: $M_{k}=M_{k+1}$,
\item Arithmetic reward sequence: $M_{k}-M_{k+1}=\gamma$, $\gamma>0$,
\item Geometric reward sequence: $M_{k+1}=\eta M_{k}$, $0\leq\eta\leq 1$,
\end{itemize}
where $\gamma$ and $\eta$ are constants.

In the next section, we show the effects of different reward structures on the allocation of CPU powers by the workers.

\section{Simulation Results}

In this section, we evaluate the all-pay auction mechanism. Table \ref{tab:simulation} summarizes the simulation parameters. With normalization, we consider the total amount of reward $\sigma$ to be $1$, i.e., $\sigma=\sum_{k=1}^{K}M_{k}=1$. We also set $m=n=2$ (see ``Task Allocation'' step in Section~\ref{sec:system}).
\begin{table}[h]
\caption{Simulation Parameters.} 
\label{tab:simulation}
\centering
\begin{tabular}{p{5.5cm} | p{2.2cm} }
\hline \hline
\textbf{Parameter}& \textbf{Values}\\ [0.5ex]
 \hline 
Unit cost of computational energy, $\theta$ & 1\\
Effective switch coefficient, $\kappa$ \cite{hao2018energy} & $10^{-25}$\\
Total number of CPU cycles required, $a$ & $5\times 10^{12}$\\
Valuation of worker $i$, $v_{i}$ & $\sim U[0,1]$\\

\hline 
\end{tabular}
\end{table}

\subsection{Monotonic Behaviour of Workers}

In the simulations, we consider a uniform distribution of the workers' valuation for the rewards, where $v_i \in [0,1]$ which are independently drawn from $F(v)=v$. From Figs.~\ref{fig:onereward}-\ref{fig:diffrewards}, it can be observed that the worker CPU power allocation increases monotonically with its valuation. Specifically, the higher the valuation of the workers for the rewards, the larger the amount of CPU power allocated for the computation subtasks. Since the workers are symmetric where their valuations are drawn from the same distribution, the workers with the same valuation contribute the same amount of CPU power.

\begin{figure*}
\centering
\begin{multicols}{2}
\includegraphics[width=0.7\columnwidth]{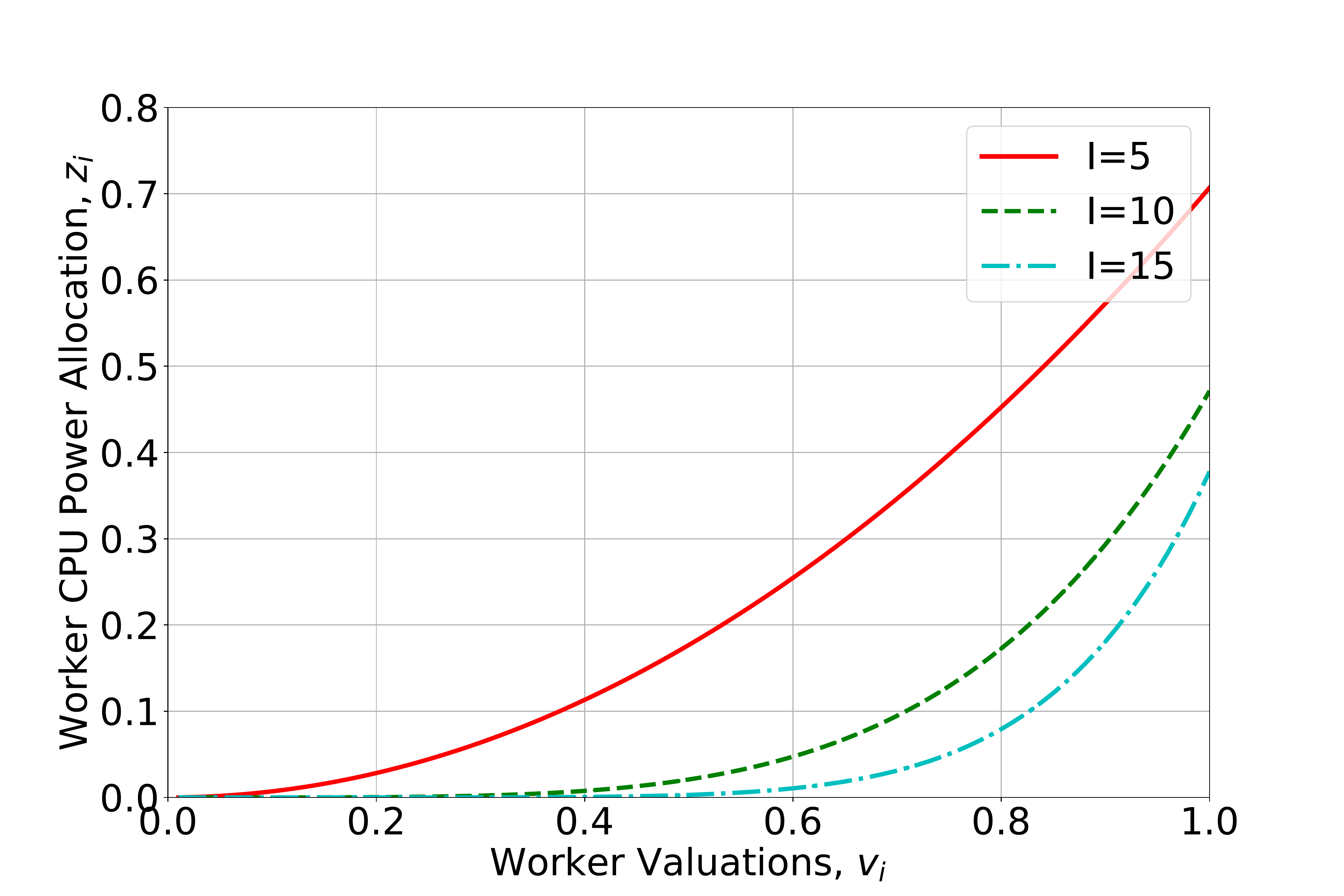}
\caption{\footnotesize{Only one reward is offered to the worker with the largest CPU power allocation.}}
\label{fig:onereward}

\includegraphics[width=0.7\columnwidth]{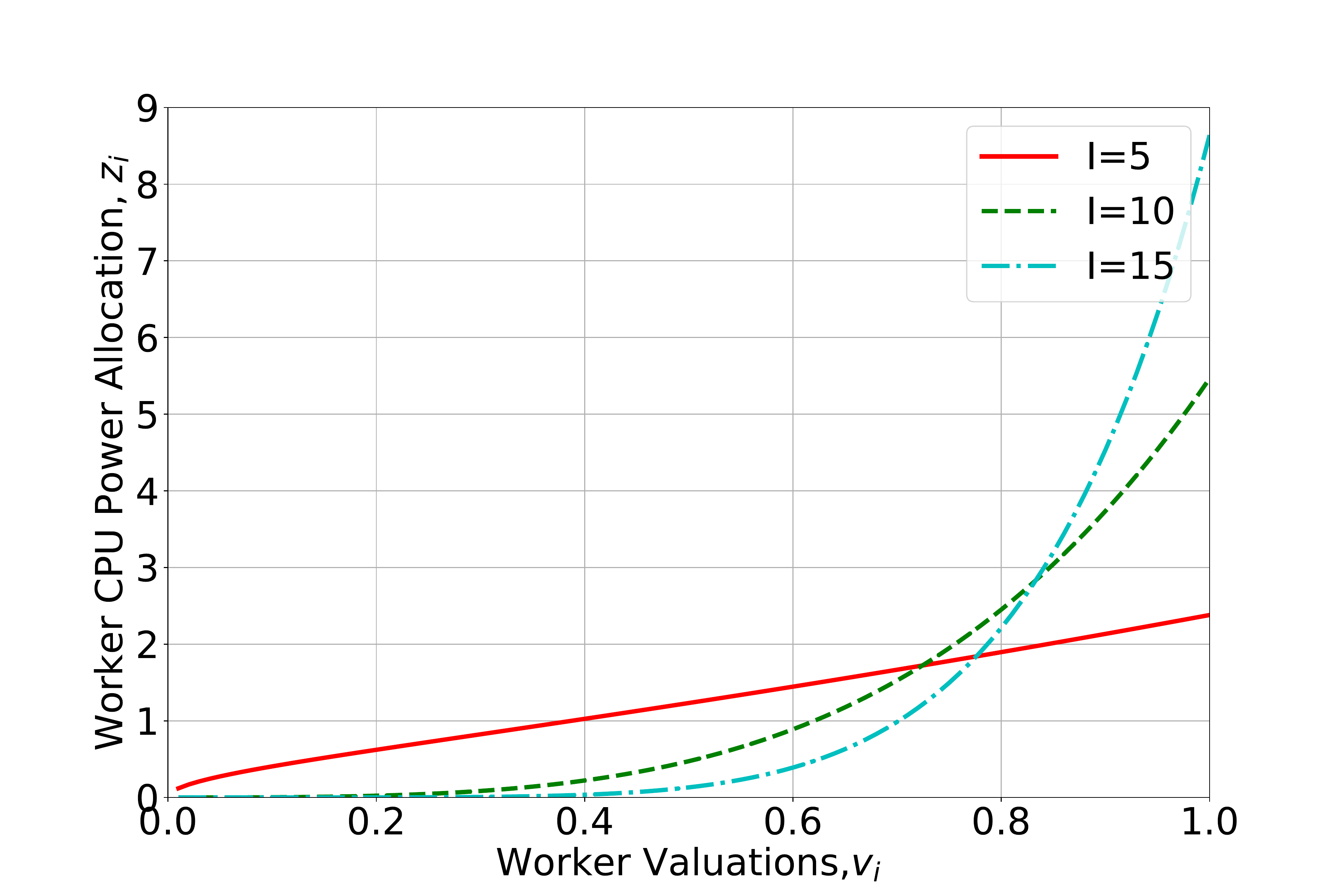}
    \caption{\footnotesize{The difference between consecutive reward amounts is by a factor of 0.8, $M_{k+1}=0.8M_{k}$.}}
    \label{fig:multigeom}

\end{multicols}
\end{figure*}

\subsection{Winner-take-all}

Based on the different reward structure adopted by the master, the workers allocate their CPU power accordingly. Figure~\ref{fig:onereward} shows that when there are 5 workers and only one reward is offered to the worker that allocates the largest amount of CPU power, the worker with the highest valuation of 1, i.e., $v_{i}=1$, is only willing to contribute $0.71$W of CPU power. However, when the master offers multiple rewards, the worker with the same valuation of 1 is willing to contribute as high as $2.38$W, $2.42$W and $2.15$W as shown in Fig.~\ref{fig:multigeom}, Fig.~\ref{fig:multiarith0.05} and Fig.~\ref{fig:multiarith0.1} respectively. With more rewards, the workers have higher chance of winning one of the rewards. Hence, to incentivize the workers to allocate for CPU power for the computation subtasks, the master is better off offering multiple rewards than a single reward.

\subsection{Multiple Rewards:}

\begin{figure*}
\centering
\begin{multicols}{3}
\includegraphics[width=\columnwidth]{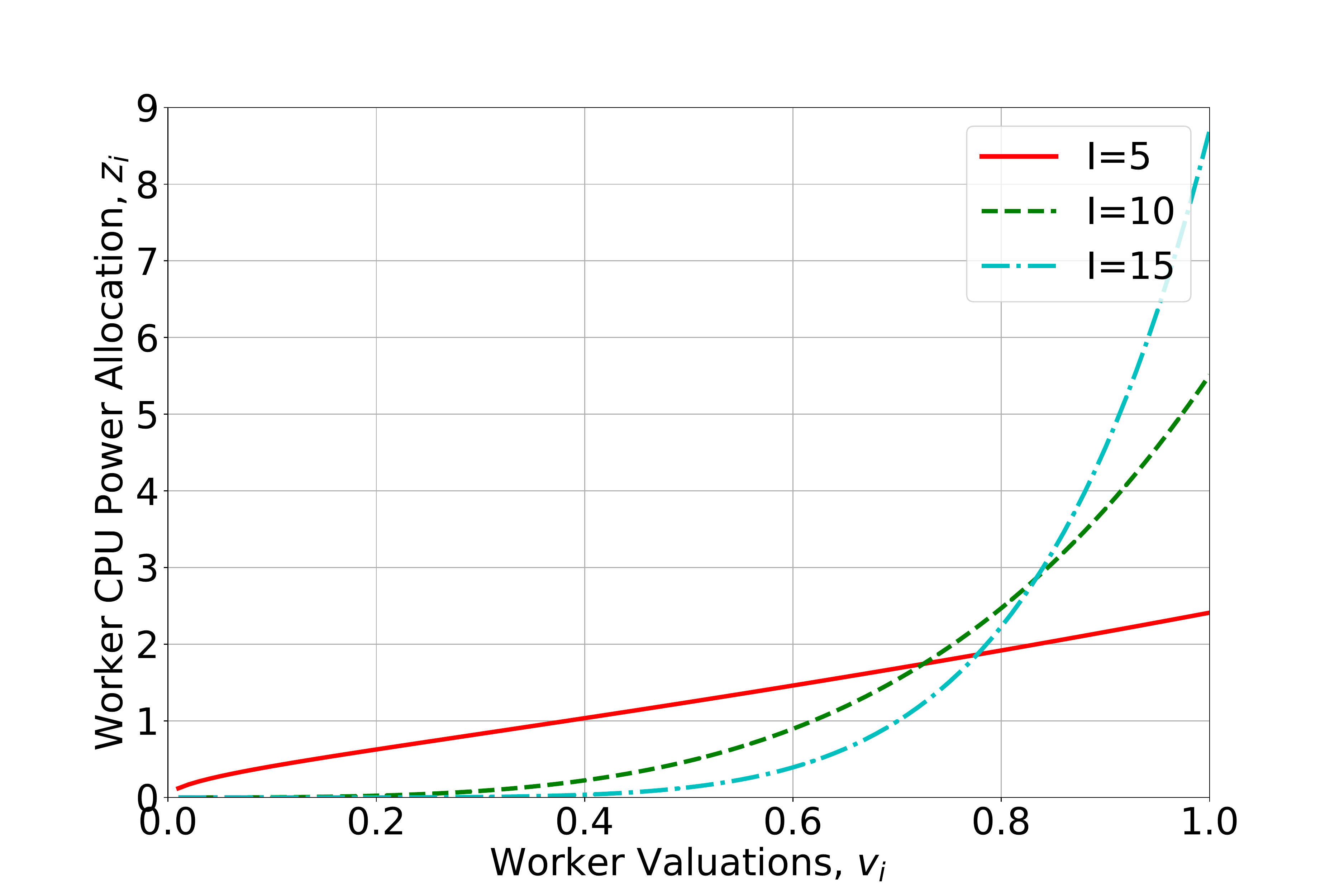} 
	\caption{\footnotesize{The difference $M_{k}-M_{k+1}$ between consecutive reward amounts is 0.05.}}
	\label{fig:multiarith0.05}
\includegraphics[width=\columnwidth]{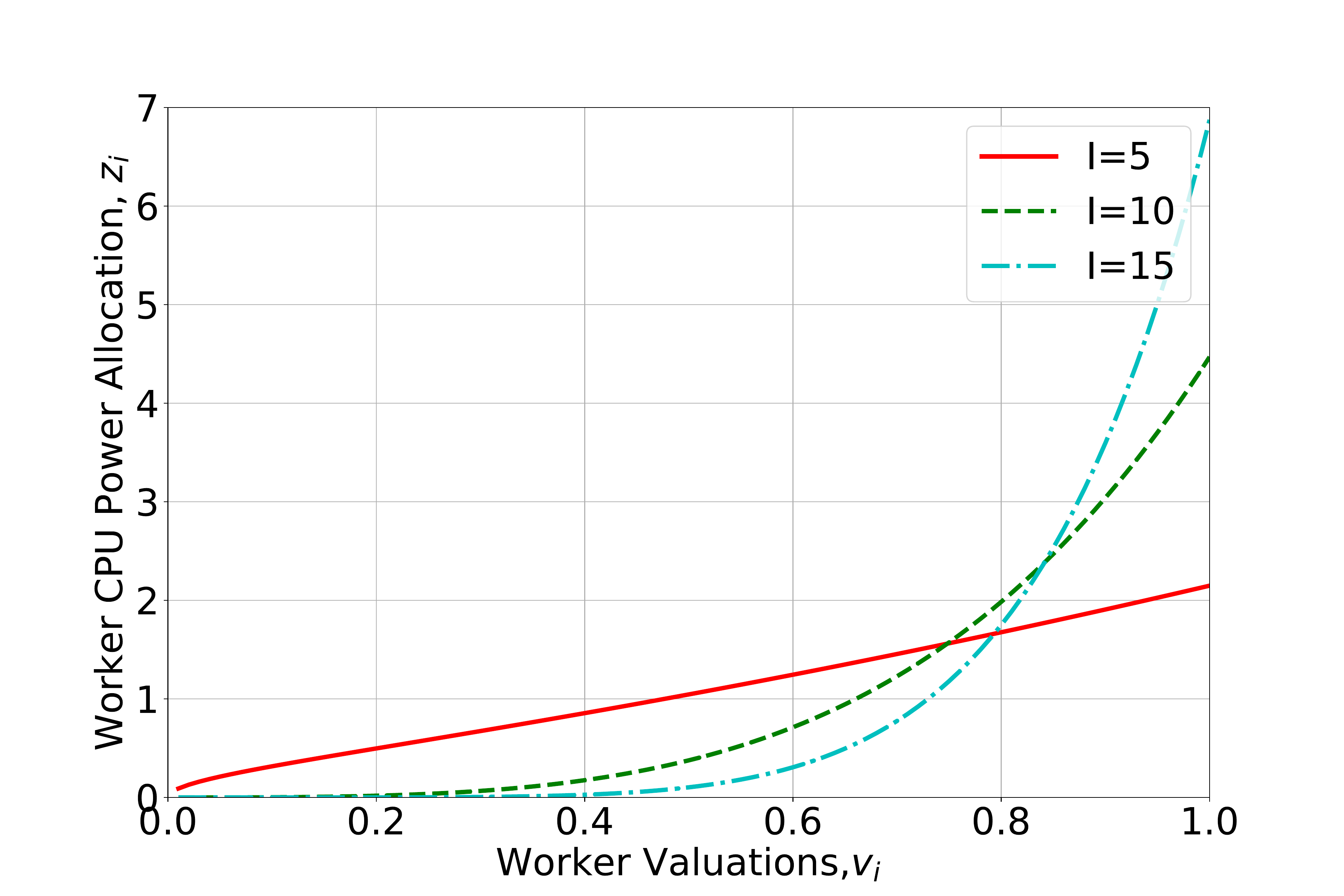} 
    	\caption{\footnotesize{The difference $M_{k}-M_{k+1}$ between consecutive reward amounts is 0.1.}}
    	\label{fig:multiarith0.1} 
\includegraphics[width=\columnwidth]{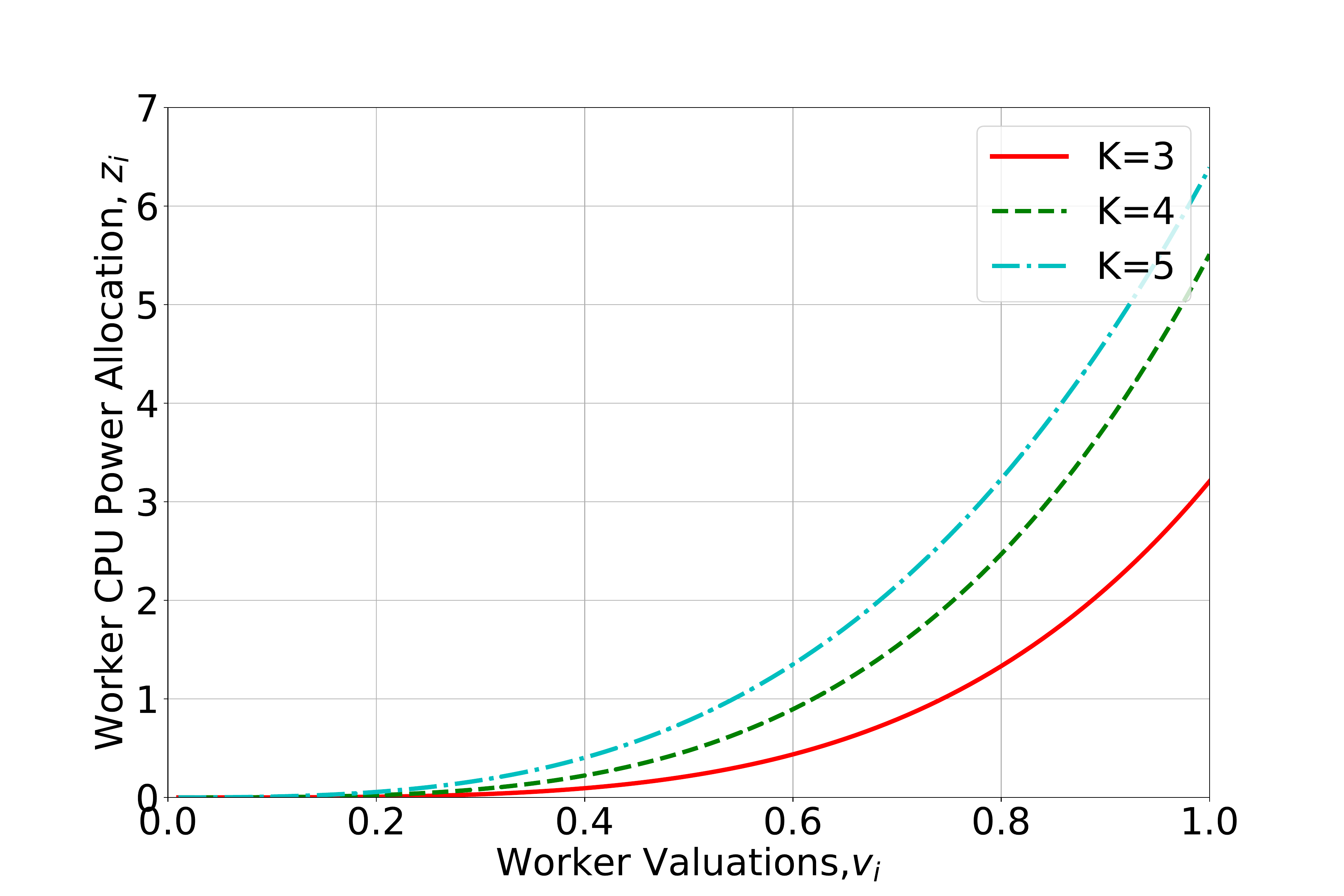}
\caption{\footnotesize{Different number $K$ of rewards, the difference between the consecutive rewards is 0.05, $I=10$.}}
\label{fig:diffrewards}

\end{multicols}
\end{figure*}



Figure~\ref{fig:multigeom}, Fig.~\ref{fig:multiarith0.05} and Fig.~\ref{fig:multiarith0.1} show the allocation of CPU power by the workers under arithmetic and geometric reward sequences. When the difference between the consecutive rewards is smaller, the workers are willing to allocate more CPU power. For example, when the difference between the consecutive rewards is 0.05, i.e., $M_{k}-M_{k+1}=0.05$, $k=1,2,\ldots,K-1$, the worker with valuation of 1 allocates $8.7$W when there are 15 workers competing for 4 rewards. However, under the same setting of 15 workers competing for 4 rewards, the worker with valuation of 1 is only willing to allocate $6.9$W when the difference between the consecutive rewards is 0.1.

\subsection{Effects of Different System Parameters}
Apart from the different reward structures, the workers also behave differently under different system parameter values, e.g., the number of workers and the number of rewards.

\subsubsection{More Workers}

When there is only one reward that is offered to the worker that allocates the largest amount of CPU power, the workers allocate more CPU power when there are 5 workers than that of 15 workers. For example, in Fig.~\ref{fig:onereward}, the worker with a valuation of $0.8$ allocates $0.45$W when there are 5 workers but only allocates $0.08$W when there are 15 workers. When there are more workers participating in the auction, the competition among the workers is stiffer and the probability of winning the reward decreases. As a result, the workers allocate a smaller amount of CPU power.

However, similar trends are only observed for workers with low valuations, e.g., $v_{i}=0.6$, when multiple rewards are offered. When the number of workers increases, the workers with low valuations reduce their allocation of CPU power for the computation subtasks. However, this is not observed for workers with high valuations, e.g., $v_{i}=0.9$. Figure~\ref{fig:multigeom}, Fig.~\ref{fig:multiarith0.05} and Fig.~\ref{fig:multiarith0.1} show that the workers with high valuations allocate more CPU power when there are more workers competing for the multiple rewards offered by the master. Specifically, when the master offers 4 rewards with a difference of 0.05 between the consecutive rewards, the worker with a valuation of $0.9$ allocates $2.16$W when there are 5 workers but allocates $4.16$W when there are 15 workers. When multiple rewards are offered, since it is possible for the workers to still win one of the remaining rewards even if they do not win the top reward, the workers are more willing to allocate their CPU power for the computation subtasks. 
Hence, the workers with high valuations allocate more CPU power to increase their chance of winning the top reward. 

\subsubsection{More Rewards} 

When the number of workers participating in the all-pay auction is fixed, the workers allocate more CPU power when there are more rewards that are offered. It is seen from Fig.~\ref{fig:diffrewards} that when there are 10 workers in the all-pay auction, the worker with a valuation of $0.8$ allocates CPU power of $3.2$W when 5 rewards are offered as compared to $1.3$W and $2.5$W when 3 and 4 rewards are offered respectively.

\section{Conclusion}
In this paper, we proposed an all-pay auction to incentivize workers to participate in the coded distributed computation tasks. Firstly, we use polynomial codes to determine the number of rewards to be offered in the all-pay auction. Then, the master determines the reward structure to maximize its utility given the strategies of the workers. For our future work, we can consider workers with valuations chosen from different probability distributions.

\bibliographystyle{ieeetr}
\bibliography{cdc_allpay-arxiv}

\end{document}